\documentclass[11pt]{article}
\pdfoutput=1

\usepackage[a4paper,margin=1in]{geometry}
\usepackage{titling}
\usepackage{float}
\usepackage[bookmarks,bookmarksopen,bookmarksdepth=2]{hyperref}
\usepackage{xspace}
\usepackage{comment}
\usepackage{mathtools}
\usepackage{mathrsfs}
\usepackage{amsmath}
\usepackage{amsfonts}
\usepackage{amssymb}
\usepackage{amsthm}
\usepackage{stmaryrd}
\usepackage{pdfsync}
\usepackage{thmtools,thm-restate}
\usepackage[textsize=small]{todonotes}
\usepackage{paralist}
\usepackage{algorithm}
\usepackage[noend]{algpseudocode}
\usepackage{fancybox}
\usepackage{mathdots}
\usepackage[utf8]{inputenc}
\usepackage[numbers]{natbib}

\declaretheorem{proposition}
\declaretheorem[sibling=proposition]{lemma}
\declaretheorem[sibling=proposition]{theorem}
\declaretheorem[sibling=proposition]{corollary}
\declaretheorem[style=remark]{example}
\declaretheorem[style=remark]{claim}

\newcommand{\expspace}{{\sc ExpSpace}\xspace}

\newcommand{\tower}{{\sc Tower}\xspace}

\newcommand{\inc}[1]{\add{#1}{1}}
\newcommand{\dec}[1]{\sub{#1}{1}}
\newcommand{\coreadd}[2]{#1 \,\, +\!\!= \, #2}
\newcommand{\coresub}[2]{#1 \,\, -\!\!= \, #2}
\newcommand{\add}[2]{$\coreadd{#1}{#2}$}
\newcommand{\sub}[2]{$\coresub{#1}{#2}$}

\newcommand{\para}[1]{\paragraph{#1.}}
\makeatletter
\def\BState{\State\hskip-\ALG@thistlm}
\makeatother

\newcommand{\Loopatmost}[1]{\Loop\ \textbf{at most $#1$ times}}

\newcommand{\comp}[2]{#1 \rhd #2}

\algrenewcommand{\algorithmiccomment}[1]{\qquad$\rightarrow$ #1}
\newcommand{\goto}[2]{\textbf{goto} {\footnotesize #1} \textbf{or} {\footnotesize #2}}
\newcommand{\gotod}[1]{\textbf{goto} {\footnotesize #1}}
\newcommand{\testz}[1]{\textbf{zero?}~$#1$}
\newcommand{\testm}[1]{\textbf{max?}~$#1$}
\newcommand{\halt}{\textbf{halt}}
\newcommand{\haltz}[1]{{\halt} \textbf{if} $#1 = 0$}
\newcommand{\vr}[1]{\mathsf{#1}}
\newcommand{\tuple}[1]{\langle #1 \rangle}

\newcommand{\setup}[1]{\textbf{setup} $#1$}
\newcommand{\ismax}[1]{\textbf{ismax} $#1$}
\newcommand{\iszero}[1]{\textbf{iszero} $#1$}
\newcommand{\reset}[1]{\textbf{reset} $#1$}

\title{The Reachability Problem for Petri Nets is Not Elementary%
\thanks{This research has been supported by
the ERC project `Lipa' within the EU Horizon 2020 research and innovation programme (No.~683080),
NCN grants 'Separation problems in automata theory' (2016/21/D/ST6/01376) and `Automatic analysis of concurrent systems' (2017/27/B/ST6/02093),
ANR programmes IdEx Bordeaux (ANR-10-IDEX-03-02) and BraVAS (ANR-17-CE40-0028), and 
the Leverhulme Trust Research Fellowship `Petri Net Reachability Conjecture' (RF-2017-579).}}

\author{%
Wojciech Czerwi\'nski \\ University of Warsaw           \\ \texttt{wczerwin@mimuw.edu.pl} \and
S{\l}awomir Lasota    \\ University of Warsaw           \\ \texttt{sl@mimuw.edu.pl} \and
Ranko Lazi\'c         \\ University of Warwick          \\ \texttt{R.S.Lazic@warwick.ac.uk} \and
J\'er\^ome Leroux     \\ CNRS \& University of Bordeaux \\ \texttt{jerome.leroux@labri.fr} \and
Filip Mazowiecki      \\ University of Bordeaux         \\ \texttt{filip.mazowiecki@u-bordeaux.fr}}

\date{}

\begin{document}

\pagenumbering{roman}
\begin{titlingpage}

\maketitle

\begin{abstract}
Petri nets, also known as vector addition systems, are a long established model of concurrency with extensive applications in modelling and analysis of hardware, software and database systems, as well as chemical, biological and business processes.  The central algorithmic problem for Petri nets is reachability: whether from the given initial configuration there exists a sequence of valid execution steps that reaches the given final configuration.
The complexity of the problem has remained unsettled since the 1960s, and it is one of the most prominent open questions in the theory of verification.  Decidability was proved by Mayr in his seminal STOC 1981 work, and the currently best published upper bound is non-primitive recursive Ackermannian of Leroux and Schmitz from LICS 2019.
We establish a non-elementary lower bound, i.e.\ that the reachability problem needs a tower of exponentials of time and space.  Until this work, the best lower bound has been exponential space, due to Lipton in 1976.  The new lower bound is a major breakthrough for several reasons.  Firstly, it shows that the reachability problem is much harder than the coverability (i.e., state reachability) problem, which is also ubiquitous but has been known to be complete for exponential space since the late 1970s.  Secondly, it implies that a plethora of problems from formal languages, logic, concurrent systems, process calculi and other areas, that are known to admit reductions from the Petri nets reachability problem, are also not elementary.  Thirdly, it makes obsolete the currently best lower bounds for the reachability problems for two key extensions of Petri nets: with branching and with a pushdown stack.

At the heart of our proof is a novel gadget so called the factorial
amplifier that, assuming availability of counters that are zero
testable and bounded by~$k$, guarantees to produce arbitrarily large
pairs of values whose ratio is exactly the factorial of~$k$.  We also
develop a novel construction that uses arbitrarily large pairs of
values with ratio $R$ to provide zero testable counters that are
bounded by~$R$.  Repeatedly composing the factorial amplifier with itself by means of the construction then enables us to compute in linear time Petri nets that simulate Minsky machines whose counters are bounded by a tower of exponentials, which yields the non-elementary lower bound.  By refining this scheme further, we in fact establish hardness for $h$-exponential space already for Petri nets with $h + 13$ counters.

\end{abstract}

\end{titlingpage}
\pagenumbering{arabic}

\section{Introduction}

Petri nets~\cite{Petri62}, also known as vector addition systems \cite{KarpM69}, \cite[cf.\ Section~5.1]{Greibach78a}, \cite{HopcroftP79}, are a long established model of concurrency with extensive applications in modelling and analysis of hardware \cite{BurnsKY00,LerouxAG15}, software \cite{GermanS92,BouajjaniE13,KKW14} and database \cite{BojanczykDMSS11,BojanczykMSS09} systems, as well as chemical~\cite{AngeliLS11}, biological \cite{PelegRA05,BaldanCMS10} and business \cite{Aalst15,LiDV17} processes (the references on applications are illustrative).  The central algorithmic problem for Petri nets is reachability: whether from the given initial configuration there exists a sequence of valid execution steps that reaches the given final configuration.

There are several presentations of Petri nets, and a number of variants of their reachability problem, all of which are equivalent.  One simple way to state the problem is: given a finite set~$T$ of integer vectors in $d$-dimensional space and two $d$-dimensional vectors $\mathbf{v}$ and $\mathbf{w}$ of nonnegative integers, does there exist a walk from $\mathbf{v}$ to $\mathbf{w}$ such that it stays within the nonnegative orthant, and its every step modifies the current position by adding some vector from~$T$?

\para{Brief History of the Problem}

Over the past half century, the complexity of the Petri nets
reachability problem has remained unsettled.  The late 1970s and the
early 1980s saw the initial burst of activity.  After an incomplete
proof by Sacerdote and Tenney~\cite{SacerdoteT77}, decidability of the
problem was established by Mayr~\cite{Mayr81,Mayr84}, whose proof was
then simplified by Kosaraju~\cite{Kosaraju82}.  Building on the
further refinements made by Lambert in the 1990s~\cite{Lambert92},
there has been substantial progress over the past ten years
\cite{Leroux10,Leroux11,Leroux12}, culminating in the first upper
bound on the complexity~\cite{LerouxS15}, recently improved to
Ackermannian~\cite{Schmitz18cigar}. 

In contrast to the progress on refining the proof of decidability and obtaining an upper bound on the complexity, Lipton's landmark result that the Petri nets reachability problem requires exponential space~\cite{lipton76} has remained the state of the art on lower bounds for over 40 years.  Moreover, in conjunction with an apparent tightness of Lipton's construction, this has led to the conjecture that the problem is \expspace-complete becoming common in the community.\footnote{For an interesting post by Lipton about his exponential space hardness result, we refer the reader to \url{https://rjlipton.wordpress.com/2009/04/08/an-expspace-lower-bound/}.}

\para{Main Result and Its Significance}

We show that the Petri nets reachability problem is not elementary, more precisely that it is hard for the class {\tower} of all decision problems that are solvable in time or space bounded by a tower of exponentials whose height is an elementary function of the input size \cite[Section~2.3]{Schmitz16toct}.  We see this result as important for several reasons:
\begin{itemize}
\item
It refutes the conjecture of \expspace-completeness, establishing that the reachability problem is much harder than the coverability (i.e., state reachability) problem; the latter is also ubiquitous but has been known to be \expspace-complete since the late 1970s~\cite{lipton76,Rackoff78}.
\item
It narrows significantly the gap to the best known upper bound
~\cite{Schmitz18cigar} in terms of the Ackermannian function, which is among the slowest-growing functions that dominate all primitive recursive functions.
\item
It implies that a plethora of problems from formal languages \cite{Crespi-ReghizziM77}, logic \cite{Kanovich95,DemriFP16,DeckerHLT14,ColcombetM14}, concurrent systems~\cite{GantyM12,EsparzaGLM17}, process calculi~\cite{Meyer09}, linear algebra~\cite{HL18} and other areas (the references are again illustrative), that are known to admit reductions from the Petri nets reachability problem, are also not elementary; for more such problems and a wider discussion, we refer to Schmitz's recent survey~\cite{Schmitz16siglog}.
\item
It makes obsolete the {\tower} lower bounds for the reachability problems for two key extensions of Petri nets: branching vector addition systems~\cite{LazicS15} and pushdown vector addition systems~\cite{LazicT17}.
\end{itemize}

\para{Petri Nets and Exponential Space Hardness}

Before we present the main ideas involved in the proof of the non-elementary lower bound for the reachability problem, let us introduce some key aspects of Petri nets by recalling the crux of Lipton's construction for the exponential space hardness.

Minsky machines, which can be thought of as deterministic finite-state machines equipped with several registers, are one of the classical universal models of computation \cite[Chapter~14]{minsky1967computation}.  The registers, which are called counters, store natural numbers (initially~$0$) and can be manipulated by only two simple operations: increments (\inc{\vr{x}}), and conditionals that either jump if a counter is zero or decrement it otherwise (\textbf{if} $\vr{x} = 0$ \textbf{then goto} {\footnotesize $L$} \textbf{else} \dec{\vr{x}}).  With appropriate restrictions, the halting problem for Minsky machines is complete for various time and space complexity classes.  Lipton's proof proceeds by reducing from the following \expspace-complete problem (cf.\ \cite[Theorems 3.1 and 4.3]{FischerMR68}): given a Minsky machine of size~$n$ with $3$ counters, does it halt after a run in which the counters remain bounded by $2^{2^n}$?

Petri nets can be construed as similar to Minsky machines, but with two important differences.  Firstly, Petri nets can increment a counter always, and can decrement a counter if positive, but cannot test whether a counter is zero.  Secondly, Petri nets are nondeterministic.  Thus, a decrement of a counter either succeeds and the run continues (if the counter was positive), or fails and the current nondeterministic branch is blocked (if the counter was zero).  It is the lack of zero tests that makes decidable~\cite{Mayr84} the reachability problem: given a Petri net and a subset of its counters, does it halt in a configuration where all the counters from the subset are zero?

To construct a Petri net that simulates the given Minsky machine of size~$n$ as long as its $3$~counters are bounded by~$2^{2^n}$, the main task is therefore checking that such a counter~$\vr{x}$ is zero.  Lipton observed that it suffices to introduce a counter~$\hat{\vr{x}}$, set up and maintain the invariant $\vr{x} + \hat{\vr{x}} = 2^{2^n}$, and implement a macro \textbf{Dec}$_n$~$\hat{\vr{x}}$ that decrements $\hat{\vr{x}}$ and increments $\vr{x}$ (i.e., performs the code \dec{\hat{\vr{x}}}\hspace{1em}\inc{\vr{x}}) exactly $2^{2^n}$ times.  That is because the code \textbf{Dec}$_n$~$\hat{\vr{x}}$\hspace{1em}\textbf{Dec}$_n$~$\vr{x}$ (where, in the latter instance of the macro, $\vr{x}$ and $\hat{\vr{x}}$ are swapped) then checks that counter~$\vr{x}$ is zero: it either succeeds and leaves $\vr{x}$ and $\hat{\vr{x}}$ unchanged if $\vr{x}$ was zero (i.e.\ $\hat{\vr{x}}$ was~$2^{2^n}$), or fails otherwise.

Lipton's construction meets that goal inductively, by setting up pairs of counters such that $\vr{x}_i + \hat{\vr{x}}_i = 2^{2^i} = \vr{y}_i + \hat{\vr{y}}_i$ and implementing a macro \textbf{Dec}$_i$ that decrements a counter and increments its complement exactly $2^{2^i}$ times, for $i = 0, 1, \ldots, n$.  Doing it for $i = 0$ is easy.  To step from $i$ to $i + 1$, consider the following code, where the loops are repeated nondeterministic numbers of times:
\begin{quote}
\begin{algorithmic}
\Loop
  \State \inc{\vr{x}_i} \quad \dec{\hat{\vr{x}}_i}
  \Loop
    \State \inc{\vr{y}_i} \quad \dec{\hat{\vr{y}}_i}
    \State \dec{\hat{\vr{x}}_{i + 1}} \quad \inc{\vr{x}_{i + 1}}
  \EndLoop
  \State \textbf{Dec}$_i$~$\vr{y}_i$
\EndLoop
\State \textbf{Dec}$_i$~$\vr{x}_i$.
\end{algorithmic}
\end{quote}
Assuming that $\vr{x}_i$ and $\vr{y}_i$ are zero at the start, there is a unique nondeterministic branch that runs the code completely (without getting blocked): it repeats the outer loop $2^{2^i}$ times with the final \textbf{Dec}$_i$~$\vr{x}_i$ both checking that $\vr{x}_i$ equals $2^{2^i}$ (i.e.\ $\hat{\vr{x}}_i$ equals zero) and resetting it to zero, and in each iteration similarly the inner loop is repeated also $2^{2^i}$ times.  Hence, by squaring $2^{2^i}$, the code decrements the counter~$\hat{\vr{x}}_{i + 1}$ and increments the counter~$\vr{x}_{i + 1}$ exactly $2^{2^{i + 1}}$ times, as required for an implementation of \textbf{Dec}$_{i + 1}$~$\hat{\vr{x}}_{i + 1}$.

\newif\iffilip
\filiptrue

\iffilip

\para{Obtaining the {\tower} Lower Bound}

To prove that the Petri nets reachability problem requires a tower of exponentials of time and space, we have to tackle two major obstacles:
\begin{enumerate}
\item
The fact that Lipton's construction applies also to the coverability problem, which has an {\expspace} upper bound~\cite{Rackoff78}, means that a construction that achieves a lower bound beyond {\expspace} cannot follow the same pattern.  For example, we cannot hope to implement a macro whose unique complete execution performs some given counter operations exactly a triply exponential number of times.
\item
It has been known for many years how Petri nets can compute various
functions weakly, in the sense that the result may be
nondeterministically either correct or
smaller~\cite{MayrM81,LerouxS14}\footnote{In their article, Mayr and
  Meyer establish that the containment problem between finite sets of
  reachable configurations of two given Petri nets is `the first
  uncontrived decidable problem which is not primitive recursive'.}.
Most notably, for all natural numbers~$n$, Grzegorczyk's
function~\cite{lob1970hierarchies} $F_n$ is computable weakly by a Petri net of size~$O(n)$.  However, even supposing that we have means of simulating zero tests of several counters bounded by some~$k$, it has been unknown how to compute exactly a value exponential in~$k$ without using $\Omega(k)$ extra counters.
\end{enumerate}

To overcome the {\expspace} barrier, we devise a novel construction for simulating zero tests of counters bounded by some~$R$: instead of relying on an ability to repeat some counter operations exactly $R$ times, it assumes that a pair of counters have been set to sufficiently large values whose ratio is exactly $R$, and it ensures that the simulations of zero tests are correct by testing that one of the two auxiliary counters is zero in the final configuration of the reachability problem instance.

In overcoming the second obstacle, surprisingly a central role is
played by the simple identity $\prod_{i = 1}^{k-1} (i+1)/i = k$.  We
devise a gadget, so called the factorial amplifier, that sets two
counters $\vr{c}$ and $\vr{d}$ to arbitrarily large values such that
$\vr{d} = \vr{c} \cdot k!$ as follows.  After initialising both
counters to a same value, the main loop uses some extra machinery and
a constant number of auxiliary counters to attempt to multiply
$\vr{c}$ and $\vr{d}$ by each of the fractions $1/i$ and $(i+1)/i$
(respectively) for $i = 1, \ldots, k-1$.  Since the multiplications
are implemented by repeated additions and subtractions, and since the
factorial amplifier cannot zero-test counters that are not bounded by~$k$, we have that the resulting values of $\vr{c}$ and $\vr{d}$ are not necessarily correct.  Nevertheless, the construction (and here an appropriate intertwining of the operations on $\vr{c}$ and $\vr{d}$ in the main loop is key) is such that the
computation is correct if and only if
the final value of $\vr{d}$ is at least (and thus exactly) $k$ times the initialised one.  Then $\vr{c}$ has necessarily been divided by $(k-1)!$, yielding the ratio $k!$ between counters $\vr{c}$ and $\vr{d}$ as required.

\else

\para{Overview of the {\tower} Lower Bound}

We reduce from the halting problem for Minsky machines of size~$n$ where the bound on the counters is a tower of exponentials.  Instead of exponentials of~$2$, we iterate the factorial operation from~$3$, i.e.\ the bound is $3!^n = 3\overbrace{! \cdots !}^n$.

How can we simulate zero tests of a counter~$\vr{x}$ bounded by $B = 3!^n$ on Petri nets?  There is a substantial obstacle: we cannot hope to implement a macro as before that increments a counter and decrements another one exactly $B$~times, because succeeding in doing so would enable us to reduce also to the coverability problem (as is the case with Lipton's construction), which would contradict the fact that the latter problem has an {\expspace} upper bound~\cite{Rackoff78}.  To overcome the obstacle, we employ the following novel idea.  Assume we have two auxiliary counters $\vr{c}$ and $\vr{d}$ that have been initialised to sufficently large values and such that $\vr{d} = \vr{c} \cdot B$.  We also use an auxiliary counter~$\hat{\vr{x}}$ such that $\vr{x} + \hat{\vr{x}} \leq B$, and employ the following code to check that $\vr{x}$ is zero:
\begin{quote}
\begin{algorithmic}
\Loop
  \State \inc{\vr{x}} \quad \dec{\hat{\vr{x}}}
  \State \dec{\vr{d}}
\EndLoop
\State \dec{\vr{c}}
\Loop
  \State \dec{\vr{x}} \quad \inc{\hat{\vr{x}}}
  \State \dec{\vr{d}}
\EndLoop
\State \dec{\vr{c}}.
\end{algorithmic}
\end{quote}
We then require that $\vr{d}$ is zero when the Petri net halts: the only way for the counter~$\vr{d}$ to reach zero at the end of the run is to maintain exactly the ratio $B$ between $\vr{c}$ and~$\vr{d}$.  Consequently, each execution of the code has to repeat both loops $B$~times, which checks that $\vr{x} = 0$ and $\hat{\vr{x}} = B$, and restores their values.  Requiring that $\vr{d}$ is zero at the end of the run is possible since we are reducing to the reachability (rather than the coverability) problem.

We are thus left with a new question: how can we set up a pair of Petri net counters that are arbitrarily large and have ratio $3!^n$ exactly?  We remark that it has been known for many years how to multiply a Petri net counter by a big (up to Ackermannian size) constant weakly~\cite{MayrM81,LerouxS14}\footnote{In their article, Mayr and Meyer establish that the containment problem between finite sets of reachable configurations of two given Petri nets is `the first uncontrived decidable problem which is not primitive recursive'.}, in the sense that the result may be nondeterministically either correct or smaller, but not how to do it exactly.  Our answer is inductive, or in other words we suppose that the former question with one fewer factorial has been answered.  More precisely, assuming that we are able to simulate zero tests of counters bounded by~$k$, we show how to set up a pair of Petri nets counters that are arbitrarily large and have ratio~$k!$.

To that end, it turns out that the simple identity $\prod_{i = 1}^{k-1} (i+1)/i = k$ is useful.  We initialise a counter~$\vr{d}$ to an arbitrarily large value, and then attempt to multiply it by each of the fractions $(i+1)/i$ in sequence using another counter~$\vr{d}'$ and the following code:
\begin{quote}
\begin{algorithmic}
\Loop
  \State \sub{\vr{d}}{\vr{i}} \quad \add{\vr{d}'}{\vr{i}+1}
\EndLoop
\Loop
  \State \dec{\vr{d}'} \quad \inc{\vr{d}}.
\EndLoop
\end{algorithmic}
\end{quote}
Here $\vr{i}$ is a counter bounded by~$k$, and the two operations
involving it are macros that contain another such counter~$\vr{i}'$,
loops and zero tests of $\vr{i}$ and $\vr{i}'$ (which we are assuming
to have at our disposal).  A key insight is that, even though the
multiplications may a priori not be exact (the counters $\vr{d}$ and
$\vr{d}'$ are not bounded by~$k$, and so we have no means of zero
testing them), their accuracy can be checked at the end by verifying
that the final value of $\vr{d}$ is at least (and hence exactly) its
initialised value times~$k$ (the right-hand side of the identity).  We
then complete the factorial amplifier by a novel construction that augments the backbone of the sequence of multiplications of $\vr{d}$ by code that makes sure that another counter~$\vr{c}$ is initialised to the same value as $\vr{d}$ but is then multiplied by $1/i$ for each $i = 1, \ldots, k-1$, i.e.\ divided by $(k-1)!$.  This ensures that the original value was divisible by $(k-1)!$ and that finally $\vr{d} = \vr{c} \cdot k!$ as required.

\fi

\para{Organisation of the Paper}

After the preliminaries in Section~\ref{s:c.p}, our scheme for
simulating zero tests of bounded counters is developed in
Section~\ref{s:elim.zero}, and the factorial amplifier for setting up arbitrarily large pairs of Petri net counters with ratio $k!$ is programmed in Section~\ref{s:exp}.

In Section~\ref{s:main}, we put the pieces together to obtain the main result, and then also show how the construction can be refined to establish that, for each positive integer~$h$, we have $h$-\expspace-hardness (tower of exponentials of height~$h$) of the reachability problem already for Petri nets with $h + 13$ counters.

The last refinement (to $h+13$ counters) is mostly relegated to the appendix available online.

\section{Counter Programs}
\label{s:c.p}

Proving the main result of this paper, namely that solving the Petri nets reachability problem requires a tower of exponentials of time and space, involves some intricate programming.  For ease of presentation, instead of working directly with Petri nets or vector addition systems, our primary language will be imperative programs that operate on variables which are called counters, and that range over the naturals (i.e.\ the nonnegative integers).

To streamline the main constructions and proofs, it will be useful to allow the programs to have two types of counters:
\begin{description}
\item[tested counters]
are bounded by a fixed positive integer~$B$ and may be tested for equality with the end points of their range, i.e.\ $0$ and~$B$;
\item[untested counters]
are unbounded and the testing commands may not be applied to them.
\end{description}
We remark that the availability of the testing commands will not make counter programs more expressive than Petri nets, because the finiteness of the range of tested counters means that their values can be seen as components of net places (or of states in vector addition systems).  However, such an enumerative translation involves a blow up proportional to the bound~$B$.

Concretely, a counter program is a sequence of commands, each of which is of one of the following five kinds:
\begin{quote}
\begin{tabular}{l@{\qquad}l}
\inc{\vr{x}}                       & (increment counter~$\vr{x}$) \\
\dec{\vr{x}}                       & (decrement counter~$\vr{x}$) \\
\goto{$L$}{$L'$}                   & (jump to either line~$L$ or line~$L'$) \\
\testz{\vr{x}}                     & (continue if counter~$\vr{x}$ equals $0$), \\
\testm{\vr{x}}                     & (continue if counter~$\vr{x}$ equals $B$),
\end{tabular}
\end{quote}
except that the last command is of the form:
\begin{quote}
\begin{tabular}{l@{\qquad}l}
\haltz{\vr{x_1}, \ldots, \vr{x}_l} & (terminate provided all \\
& the listed counters are zero).
\end{tabular}
\end{quote}
Note that the two types of counters are not declared explicitly: without loss of generality, a counter~$\vr{x}$ is regarded as tested (and thus has the range $\{0, \ldots, B\}$) if and only if it occurs in a \testz{\vr{x}} or \testm{\vr{x}} command in the program.

We use a shorthand {\halt} when no counter is required to be zero at termination.

To illustrate how the available commands can be used to express further constructs, addition \add{\vr{x}}{m} and substraction \sub{\vr{x}}{m} of a natural constant $m$ can be written as $m$ consecutive increments \inc{\vr{x}} and decrements \dec{\vr{x}} (respectively).  As another illustration, conditional jumps \textbf{if} $\vr{x} = 0$ \textbf{then goto} {\footnotesize $L$} \textbf{else} \dec{\vr{x}} which feature in common definitions of Minsky machines can be written as:
\begin{quote}
\begin{algorithmic}[1]
\State \goto{\ref{l:goto.z}}{\ref{l:goto.nz}}
\State \testz{\vr{x}} \label{l:goto.z}
\State \gotod{$L$}
\State \dec{\vr{x}},                    \label{l:goto.nz}
\end{algorithmic}
\end{quote}
where \gotod{$L$} is a shorthand for the deterministic jump \goto{$L$}{$L$}.

We emphasise that counters (both tested and untested) are not permitted to have negative values.  In the example we have just seen, that is why the decrement in line~\ref{l:goto.nz} works also as a non-zero test.

Two more remarks may be useful.  Firstly, our notion of counter programs only serves as a convenient medium for presenting both Petri nets and Minsky machines with bounded counters, and the exact syntax is not important; we were inspired here by Esparza's presentation~\cite[Section~7]{Esparza98} of Lipton's lower bound~\cite{lipton76}.  Secondly, although the \haltz{\vr{x_1}, \ldots, \vr{x}_l} commands could be expressed by zero tests followed by just {\halt}, having them as atomic commands makes it possible to require untested counters to be zero at termination.  The latter feature makes untested counters correspond to Petri net counters, which are unbounded, and can be zero tested only at the start and finish of runs by specifying initial and final configurations in instances of the reachability problem.


\subsection{Runs and Computed Relations}

A \emph{$B$-run} of a program from an initial valuation of all its counters is a run in which all values of all counters are at least~$0$, all values of all tested counters are at most~$B$, and the max tests are interpreted as checks for equality with~$B$.

We say that such a run is \emph{halted} if and only if it has successfully executed its {\halt} command (which is necessarily the program's last); otherwise, the run is either \emph{partial} or \emph{infinite}.  Observe that, due to a decrement that would cause a counter to become negative, or due to an increment that would exceed the bound of a tested counter, or due to an unsuccesful zero or max test, or due to an unsuccessful terminal check for zero, a partial run may be maximal because it is blocked from further execution.  Moreover, due to nondeterministic jumps, the same program from the same initial valuation may have various $B$-runs in each of the three categories: halted runs, maximal partial runs, and infinite runs.  We are mostly going to be interested in final counter valuations that are reached by halted runs.

We regard a run as \emph{complete} if and only if it is halted and its initial valuation assigns zero to every counter. Let $\vr{x}_1, \ldots, \vr{x}_l$ be some (not necessarily all) of the counters in the program. We say that the \emph{relation $B$-computed in $\vr{x}_1, \ldots, \vr{x}_l$} by a program is the set of all tuples $\tuple{v_1, \ldots, v_l}$ such that the program has a complete $B$-run whose final valuation assigns to every counter~$\vr{x}_i$ the natural number~$v_i$.

We may consider the same program with more than one bound for its tested counters.  When the bound~$B$ is clear, or when it is not important because there are no tested counters, we may write simply `run' and `computed' instead of `$B$-run' and `$B$-computed' (respectively).

\subsection{Examples}

\begin{example}
Consider the following program, where $C$ is a natural constant, and we observe that all the counters are untested:
\begin{quote}
\begin{algorithmic}[1]
\State \add{\vr{x}'}{C}
\State \goto{\ref{l:testz.exit}}{\ref{l:testz.iter}} \label{l:testz.rep}
\State \inc{\vr{x}} \quad \dec{\vr{x}'}              \label{l:testz.iter}
\State \add{\vr{y}}{2}                               \label{l:testz.end}
\State \gotod{\ref{l:testz.rep}}
\State \haltz{\vr{x}'}.                              \label{l:testz.exit}
\end{algorithmic}
\end{quote}

It repeats the block of three commands in lines \ref{l:testz.iter}--\ref{l:testz.end} some number of times chosen nondeterministically (possibly zero, possibly infinite) and then halts provided counter~$\vr{x}'$ is zero.  Replacing the two jumps by more readable syntactic sugar, we may write this code as:
\begin{quote}
\begin{algorithmic}[1]
\State \add{\vr{x}'}{C}
\Loop
  \State \inc{\vr{x}} \quad \dec{\vr{x}'}
  \State \add{\vr{y}}{2}
\EndLoop
\State \haltz{\vr{x}'}.
\end{algorithmic}
\end{quote}

It is easy to see that there is a unique complete run (and there are no infinite runs), in which the loop is iterated exactly $C$ times.  Thus, the relation computed in $\vr{x}, \vr{y}$ is the set with the single tuple $\tuple{C,2C}$.
\qed
\end{example}

\begin{example} \label{ex:arguments}
We shall need to reason about properties of counter valuations at certain points in programs.  As an example which will be useful later for simulating tested counters by untested ones, consider a fixed positive integer~$B$ and assume that
\begin{equation}
\label{eq:B}
\vr{x} + \hat{\vr{x}} \leq B \text{ and } \vr{d} \geq \vr{c} \cdot B
\end{equation}
holds in a run at the entry to (i.e., just before executing) the program fragment
\begin{quote}
\begin{algorithmic}
\Loop
  \State \inc{\vr{x}} \quad \dec{\hat{\vr{x}}}
  \State \dec{\vr{d}}
\EndLoop
\State \dec{\vr{c}}.
\end{algorithmic}
\end{quote}

The number of times the loop has been iterated by a run that also exits (i.e., completes executing) the program fragment is nondeterministic, so let us denote it by~$K$.  It is easy to see that property~\eqref{eq:B} necessarily also holds at the exit, since:
\begin{itemize}
\item
the sum $\vr{x} + \hat{\vr{x}}$ is maintained by each iteration of the loop,
\item
we have that $K \leq B$, and
\item
counters $\vr{d}$ and $\vr{c}$ have been decreased by $K$ and $1$ (respectively).
\end{itemize}

Continuing the example, if we additionally assume that the exit counter valuation satisfies $\vr{d} = \vr{c} \cdot B$, then we deduce that:
\begin{itemize}
\item
necessarily $K = B$,
\item
$\vr{d} = \vr{c} \cdot B$ also held at the entry, and
\item
$\vr{x} = 0$ and $\hat{\vr{x}} = B$ at the entry, and their values at the exit are swapped.
\end{itemize}

We have thus seen two small arguments, one based on propagating properties of counter valuations forwards through executions of program fragments, and the other backwards.  Both kinds will feature in the sequel.
\qed
\end{example}

\subsection{Petri Nets Reachability Problem}

It is well known that Petri nets \cite{Petri62}, vector addition systems \cite{KarpM69}, and vector addition systems with states \cite[cf.\ Section~5.1]{Greibach78a}, \cite{HopcroftP79} are alternative presentations of the same model of concurrent processes, in the sense that between each pair there exist straightforward translations that run in polynomial time and preserve the reachability problem; for further details, see e.g.\ the recent survey \cite[Section~2.1]{Schmitz16siglog}.

Since counter programs without tested counters can be seen as presentations of vector addition systems with states, where the latter are required to start with all vector components zero and to finish with vector components zero as specified by the {\halt} command, the Petri nets reachability problem can be stated as:
\begin{description}
\item[Input] A counter program without tested counters.
\item[Question] Does it have a complete run?
\end{description}

We remark that restricting further to programs where no counter is required to be zero finally (i.e., where the last command is just {\halt}) turns this problem into the Petri nets \emph{coverability} problem.  In the terminology of vector addition systems with states, the latter problem is concerned with reachability of just a state, with no requirement on the final vector components.  Lipton's {\expspace} lower bound~\cite{lipton76} holds already for the coverability problem, which is in fact \expspace-complete~\cite{Rackoff78}.

\subsection{A \tower-Complete Problem}

Let us write $!^n$ for the $n^{\text{th}}$ iterate of factorial, so that $a!^n = a\overbrace{! \cdots !}^n$.

To prove that the Petri nets reachability problem is not elementary, we shall provide a linear-time reduction from the following canonical problem.  It is complete for the class {\tower} of all decision problems that are solvable in time or space bounded by a tower of exponentials whose height is an elementary function of the input size \cite[Section~2.3]{Schmitz16toct}, with respect to elementary reductions.

\begin{description}
\item[Input] A counter program of size~$n$, without untested counters.
\item[Question] Does it have a complete $3!^n$-run?
\end{description}

For confirming that this problem is \tower-complete, we refer to \cite[Section~4.1]{Schmitz16toct} and \cite[Section~4.2]{Schmitz16toct} for the robustness of the class with respect to the choices of the fast-growing function hierarchy (here based on the factorial operation) and of the computational model (here nondeterministic Minsky machines), respectively.

\section{Simulating Tests}
\label{s:elim.zero}

We now introduce our central notion of amplifier for a ratio, and define a special operator for composing them with programs.  Provided the ratio of the amplifier is the same as the bound of the program's tested counters, the resulting composition will be an equivalent program in which those counters have become untested.  That is accomplished through eliminating the original program's zero and max tests by simulating them and using the amplifier to check that the simulations are correct, where a price to pay is introducing an extra untested counter for each of the original tested ones.

The amplifiers themselves may have tested counters.  An amplifier
whose tested counters are bounded by $B$ and whose ratio is a larger
number~$R$, called a $B$-amplifier by $R$, can then be seen, in conjunction with the composition operator, as a means for transforming programs whose tested counters are bounded by $R$ into equivalent programs whose tested counters are bounded by~$B$.  In the special case when the amplifier has no tested counters, the same will be true of the resulting programs.

Another feature, which will be key in Section~\ref{s:main}, is that
more powerful amplifiers will be obtainable by composition: applying
the operator to a $B$-amplifier by $B'$, and $B'$-amplifier by $B''$,
will produce a $B$-amplifier by~$B''$.

\subsection{Construction}

Suppose that:
\begin{itemize}

\item
$B$ and $R$ are positive integers;

\item
$\mathcal{A}$ is a \emph{$B$-amplifier by~$R$}, i.e.\ a program such that the relation it $B$-computes in counters $\vr{b}, \vr{c}, \vr{d}$ is
\[\{\tuple{b, c, d} \;:\; b = R,\ c > 0,\ d = c \cdot b\};\]

\item
$\mathcal{P}$ is a program.
\end{itemize}

\begin{example} \label{ex:ratio}
As an example to be used later, when $R$ is sufficiently small to
write $R$ consecutive increments explicitly, it is very easy to code
an amplifier by~$R$:
\begin{quote}
\begin{algorithmic}[1]
\State \add{\vr{b}}{R}
  \Comment{set $\vr{b}$ to constant $R$}
  \State \inc{\vr{c}} \quad \add{\vr{d}}{R}
\Loop
  \State \inc{\vr{c}} \quad \add{\vr{d}}{R}
\EndLoop
\State \halt.
\end{algorithmic}
\end{quote}
Observe that this amplifier does not have any tested counters and so,
for every positive integer~$B$, it is a $B$-amplifier by~$R$.
\qed
\end{example}

Under the stated assumptions, we now define a construction of a program $\comp{\mathcal{A}}{\mathcal{P}}$ which $B$-computes any relation that is $R$-computed by~$\mathcal{P}$.  The idea is to turn each tested counter~$\vr{x}$ of $\mathcal{P}$ into an untested one through supplementing it by a new counter~$\hat{\vr{x}}$ and ensuring that the invariant $\vr{x} + \hat{\vr{x}} = R$ is maintained, so that zero tests of $\vr{x}$ can be replaced by loops that $R$ times increment $\vr{x}$ and then $R$ times decrement~$\vr{x}$, and similarly for max tests.  Counter~$\vr{b}$ provided by $\mathcal{A}$ is employed to initialise each complement counter~$\hat{\vr{x}}$, whereas $\vr{c}$ and $\vr{d}$ are used to ensure that if $\vr{d}$ is zero at the end of the run then all the loops in the simulations of the zero and max tests iterated $R$ times as required.  Concretely, the program $\comp{\mathcal{A}}{\mathcal{P}}$ is constructed as follows:
\begin{enumerate}[(i)]

\item\label{e:comp.1}
counters are renamed if necessary so that no counter occurs in both $\mathcal{A}$ and~$\mathcal{P}$;

\item\label{e:comp.2}
letting $\vr{x}_1$, \ldots, $\vr{x}_l$ be the tested counters of $\mathcal{P}$, new counters $\hat{\vr{x}}_1$, \ldots, $\hat{\vr{x}}_l$ are introduced and the following code is inserted at the beginning of~$\mathcal{P}$:
\begin{quote}
\begin{algorithmic}
\Loop
  \State \inc{\hat{\vr{x}}_1} \quad $\cdots$ \quad \inc{\hat{\vr{x}}_l}
  \State \dec{\vr{b}} \quad \dec{\vr{d}}
\EndLoop
\State \dec{\vr{c}}
\end{algorithmic}
\end{quote}
(we shall show that complete runs necessarily iterate this loop $R$ times, i.e.\ until counter~$\vr{b}$ becomes zero);

\item\label{e:comp.3}
every \add{\vr{x}_i}{1} command in $\mathcal{P}$ is replaced by two commands
\begin{quote}
\begin{algorithmic}
\State \add{\vr{x}_i}{1} \quad \sub{\hat{\vr{x}}_i}{1};
\end{algorithmic}
\end{quote}

\item\label{e:comp.4}
every \sub{\vr{x}_i}{1} command in $\mathcal{P}$ is replaced by two commands
\begin{quote}
\begin{algorithmic}
\State \sub{\vr{x}_i}{1} \quad \add{\hat{\vr{x}}_i}{1};
\end{algorithmic}
\end{quote}

\item\label{e:comp.5z}
every \testz{\vr{x}_i} command in $\mathcal{P}$ is replaced by the following code:
\begin{quote}
\begin{algorithmic}
\Loop
  \State \inc{\vr{x}_i} \quad \dec{\hat{\vr{x}}_i}
  \State \dec{\vr{d}}
\EndLoop
\State \dec{\vr{c}}
\Loop
  \State \dec{\vr{x}_i} \quad \inc{\hat{\vr{x}}_i}
  \State \dec{\vr{d}}
\EndLoop
\State \dec{\vr{c}}
\end{algorithmic}
\end{quote}
(we shall show that complete runs necessarily iterate each of the two loops $R$ times, i.e.\ they check that $\vr{x}_i$ equals $0$ through checking that $\hat{\vr{x}}_i$ equals $R$ by transferring $R$ from $\hat{\vr{x}}_i$ to $\vr{x}_i$ and then back);

\item\label{e:comp.5m}
every \testm{\vr{x}_i} command in $\mathcal{P}$ is replaced analogously, i.e.\ by the code as for \testz{\vr{x}_i} but with the increments and decrements of $\vr{x}_i$ and $\hat{\vr{x}}_i$ swapped;

\item\label{e:comp.6}
letting $\vr{y}_1, \ldots, \vr{y}_m$ (respectively, $\vr{z}_1, \ldots, \vr{z}_h$) be the counters that are required to be zero at termination of $\mathcal{A}$ (respectively, $\mathcal{P}$), the code of $\comp{\mathcal{A}}{\mathcal{P}}$ consists of the code of $\mathcal{A}$ concatenated with the code of $\mathcal{P}$ modified as stated, both without their {\halt} commands, and ending with the command
\begin{quote}
\begin{algorithmic}
\State \haltz{
\vr{d}, \vr{y}_1, \ldots, \vr{y}_m, \vr{z}_1, \ldots, \vr{z}_h}.
\end{algorithmic}
\end{quote}
\end{enumerate}

We remark that simulating zero tests of counters bounded by some $R$ using transfers from and to their complements is a well-known technique that can be found already in Lipton~\cite{lipton76}; the novelty here is the cumulative verification of such simulations, through decreasing appropriately the two counters $\vr{d}$ and $\vr{c}$ whose ratio is $R$, and checking that $\vr{d}$ is zero finally.

\subsection{Correctness}
\para{Correctness}

The next proposition states that the construction of $\comp{\mathcal{A}}{\mathcal{P}}$ is correct in the sense that its $B$-computed relations in counters of $\mathcal{P}$ are the same as those $R$-computed by~${\mathcal{P}}$.  (We shall treat any renamings of counters in step~(\ref{e:comp.1}) of the construction as implicit.)  In one direction, the proof proceeds by observing that $\comp{\mathcal{A}}{\mathcal{P}}$ can simulate faithfully any complete $R$-run of $\mathcal{P}$.
In the other direction we argue that although some of the loops introduced in steps (\ref{e:comp.5z}) and (\ref{e:comp.5m}) may iterate fewer than $R$ times and hence erroneously validate a test, the ways in which counters $\vr{c}$ and $\vr{d}$ are set up by $\mathcal{A}$ and used in the construction ensure that no such run can continue to a complete one.  Informally, as soon as a loop in a simulation of a test iterates fewer than $R$ times, the equality $\vr{d} = \vr{c} \cdot R$ turns into the strict inequality $\vr{d} > \vr{c} \cdot R$ which remains for the rest of the run, preventing counter~$\vr{d}$ from reaching zero.

\begin{proposition}
\label{pr:comp.corr}
For every valuation of counters of~$\mathcal{P}$, it occurs after a complete $B$-run of $\comp{\mathcal{A}}{\mathcal{P}}$ if and only if it occurs after a complete $R$-run of~$\mathcal{P}$.
\end{proposition}

\begin{proof}
The `if' direction is straightforward: from a complete $R$-run of $\mathcal{P}$ with a total of $q$ zero and max tests, obtain a complete $B$-run of $\comp{\mathcal{A}}{\mathcal{P}}$ with the same final valuation of counters of $\mathcal{P}$ by
\begin{itemize}
\item
running $\mathcal{A}$ to termination with $\vr{b} = R$, $\vr{c} = 2 q + 1$, $\vr{d} = \vr{c} \cdot R$ and all of $\vr{y}_1, \ldots, \vr{y}_m$ equal to~$0$, where the latter counters will remain untouched for the rest of the run and hence satisfy the requirement to be zero finally (cf.\ step~(\ref{e:comp.6}) of the construction),
\item
iterating the loop in step~(\ref{e:comp.2}) $R$ times to initialise each complement counter $\hat{\vr{x}}_i$ to~$R$, which also subtracts $R$ and $1$ from $\vr{d}$ and $\vr{c}$ (respectively) as well as decreases $\vr{b}$ to~$0$, and
\item
in place of every zero or max test in~$\mathcal{P}$, iterating both loops in step (\ref{e:comp.5z}) or (\ref{e:comp.5m}) (respectively) $R$ times, which subtracts $2 R$ and $2$ from $\vr{d}$ and $\vr{c}$ (again respectively), eventually decreasing them both to~$0$.
\end{itemize}

For the `only if' direction, consider a complete $B$-run of $\comp{\mathcal{A}}{\mathcal{P}}$.  Extracting from it a complete $R$-run of $\mathcal{P}$ with the same final valuation of counters of $\mathcal{P}$ is easy once we show that, for each simulation of a \testz{\vr{x}_i} or \testm{\vr{x}_i} command by the code in step (\ref{e:comp.5z}) or (\ref{e:comp.5m}) of the construction, the values of $\vr{x}_i$ at the start and at the finish of the code are $0$ or $R$ (respectively).

Firstly, by step~(\ref{e:comp.6}) and the fact that counters $\vr{y}_1, \ldots, \vr{y}_m$ are not used after executing the part of code from $\mathcal{A}$, we have that the values of $\vr{b}$, $\vr{c}$ and $\vr{d}$ that have been provided by $\mathcal{A}$ satisfy $\vr{b} = R$ and $\vr{d} = \vr{c} \cdot R$.  After the code in step~(\ref{e:comp.2}) 
we therefore have that $\vr{x}_i + \hat{\vr{x}}_i \leq R$ for all~$i$.  Recalling the reasoning in Example~\ref{ex:arguments} and arguing forwards through the run, we infer that
\[\vr{x}_i + \hat{\vr{x}}_i \leq R \text{ for all $i$, and } \vr{d} \geq \vr{c} \cdot R\]
is an invariant that is maintained by the rest of the run.

Now, due to step~(\ref{e:comp.6}) again, 
$\vr{d}$ is zero finally, and so the inequality $\vr{d} \geq \vr{c} \cdot R$ is finally an equality.  Therefore, $\vr{c}$ is zero finally as well.  Recalling again the reasoning in Example~\ref{ex:arguments} and arguing backwards through the run, we conclude that in fact $\vr{d} = \vr{c} \cdot R$ has been maintained and that, for each simulation of a \testz{\vr{x}_i} or \testm{\vr{x}_i} command, each of the two loops has been iterated exactly $R$ times, and hence the values of $\vr{x}_i$ at its start and at its finish have been as required.  Also, the loop introduced in step~(\ref{e:comp.2}) has been iterated $R$ times, and $\vr{b}$ is zero finally.
\end{proof}

\section{Factorial Amplifier}
\label{s:exp}

This section is the technical core of the paper.  It provides a single
program~$\mathcal{F}$ called the factorial amplifier which is, for any positive integer~$k$, a
$k$-amplifier by $k!$.  Together with the composition operator from
Section~\ref{s:elim.zero}, we shall then have all the tools needed for
obtaining our main result in Section~\ref{s:main}: chains of
compositions of $\mathcal{F}$ with itself will yield amplifiers by ratios which are towers of exponentials.

\subsection{A simple program}
\label{subs:simple}

As a warm up for the presentation of the main program and the proof of its correctness, let us consider a simpler program~$\mathcal{E}$ specified in Algorithm~\ref{a:E}. Two macros are used to aid readability, and we now expand them, noting that hidden within them is another counter~$\vr{i}'$:
\begin{description}
\item[\sub{\vr{x}}{\vr{i}}:]
To subtract the current value of counter~$\vr{i}$, we employ the auxiliary counter~$\vr{i}'$ to which the value of $\vr{i}$ is transferred and then transferred back.  At the start of the code, $\vr{i}'$ is assumed to be zero, and the same is guaranteed at the finish.
\begin{quote}
\begin{algorithmic}
\Loop
  \State \dec{\vr{i}} \quad \inc{\vr{i}'} \quad \dec{\vr{x}}
\EndLoop
\State \testz{\vr{i}}
\Loop
  \State \dec{\vr{i}'} \quad \inc{\vr{i}}
\EndLoop
\State \testz{\vr{i}'}
\end{algorithmic}
\end{quote}

\item[\add{\vr{x}'}{\vr{i} + 1}:]
This is very similar, except for the extra increment of~$\vr{x}'$.
\begin{quote}
\begin{algorithmic}
\State \inc{\vr{x}'}
\Loop
  \State \dec{\vr{i}} \quad \inc{\vr{i}'} \quad \inc{\vr{x}'}
\EndLoop
\State \testz{\vr{i}}
\Loop
  \State \dec{\vr{i}'} \quad \inc{\vr{i}}
\EndLoop
\State \testz{\vr{i}'}
\end{algorithmic}
\end{quote}
\end{description}

\begin{algorithm}
\caption{Counter program~$\mathcal{E}$.}
\label{a:E}
\begin{algorithmic}[1]
\Statex //\emph{Untested counters}:~$\vr{x}$, $\vr{y}$, $\vr{x}'$
\Statex //\emph{Tested counters}:~$\vr{i}$, $\vr{i}'$
\State \inc{\vr{i}} \quad
       \inc{\vr{x}} \quad \inc{\vr{y}}
\Loop
  \State \inc{\vr{x}} \quad \inc{\vr{y}}
\EndLoop
\Loop
  \Loop
    \State \sub{\vr{x}}{\vr{i}} \quad \add{\vr{x}'}{\vr{i}+1}
  \EndLoop
  \Loop
    \State \dec{\vr{x}'} \quad \inc{\vr{x}}
  \EndLoop
  \State \inc{\vr{i}}
\EndLoop
\State \testm{\vr{i}}
\Loop
  \State \sub{\vr{x}}{\vr{i}} \quad \dec{\vr{y}}
\EndLoop
\State \haltz{\vr{y}}
\end{algorithmic}
\end{algorithm}

Program~$\mathcal{E}$ has untested counters $\vr{x}$, $\vr{x}'$ and $\vr{y}$, and tested counters $\vr{i}$ and $\vr{i}'$.  Assuming that the bound for the tested counters is a positive integer~$k$, the program does the following:
\begin{itemize}
\item
initialises $\vr{x}$ and $\vr{y}$ to some positive integer~$a$ chosen nondeterministically, which will be kept unchanged in counter~$\vr{y}$ until the final loop;
\item
in each iteration of the main loop, uses counter~$\vr{x}'$ to attempt to multiply counter~$\vr{x}$ by the fraction $(\vr{i} + 1) / \vr{i}$;
\item
by the final loop and the terminal check that counter~$\vr{y}$ is zero, halts provided the value of $\vr{x}$ is at least $a \cdot k$ (in which case it will be exactly $a \cdot k$).
\end{itemize}

The first and easier part of the exercise is to show that, for any positive~$a$, there exists a complete $k$-run of $\mathcal{E}$ that initialises $\vr{x}$ and $\vr{y}$ to~$a$, then multiplies $\vr{x}$ exactly by all the fractions $(i + 1) / i$ for $i = 1, \ldots, k-1$, and finally checks that $\vr{x}$ equals $a \cdot k$.

The second part is to show the converse, i.e.\ that any complete $k$-run of $\mathcal{E}$ is of that form.  As a hint, we remark that this is the case because as soon as a multiplication of $\vr{x}$ by a fraction $(i + 1) / i$ does not complete accurately (because either the first inner loop does not decrease $\vr{x}$ to exactly zero, or the second inner loop does not decrease $\vr{x}'$ to exactly zero), it will not be possible to repair that error in the rest of the run, in the sense that the value of $\vr{x}$ at the end of the main loop will necessarily be strictly smaller than $a \cdot k$ and thus it will be impossible to complete the run.  We also remark that this vitally depends on the fact that all the fractions $(i + 1) / i$ are greater than~$1$.

\subsection{Amplifiers}

The definition of $\mathcal{F}$ in Algorithm~\ref{a:F} is presented at a high level for readability.  In addition to the two macros for subtracting $\vr{i}$ and adding $\vr{i}+1$ presented in the previous subsection,
one further macro is used:
\begin{description}
\item[\textbf{loop at most} $\vr{b}$ \textbf{times} $<${\it body}$>$:]
To express this construct, we employ the auxiliary counter $\vr{b}'$ to which the value of $\vr{b}$ is transferred and then transferred back.  Provided $\vr{b}'$ is zero at the start, the body is indeed performed at most $\vr{b}$ times.
\begin{quote}
\begin{algorithmic}
\Loop
  \State \dec{\vr{b}} \quad \inc{\vr{b}'}
\EndLoop
\Loop
  \State \dec{\vr{b}'} \quad \inc{\vr{b}}
  \State $<${\it body}$>$
\EndLoop
\end{algorithmic}
\end{quote}
\end{description}

\begin{algorithm}
\caption{Factorial Amplifier~$\mathcal{F}$.}
\label{a:F}
\begin{algorithmic}[1]
\Statex //\emph{Untested counters}:~$\vr{b}$, $\vr{b}'$, $\vr{c}$,
$\vr{c'}$, $\vr{d}$, $\vr{d}'$, $\vr{x}$, $\vr{y}$
\Statex //\emph{Tested counters}:~$\vr{i}$, $\vr{i}'$
\State \inc{\vr{i}} \quad \inc{\vr{b}} \quad
       \inc{\vr{c}} \quad \inc{\vr{d}} \quad \inc{\vr{x}} \quad \inc{\vr{y}}
\Loop
  \State \inc{\vr{c}} \quad \inc{\vr{d}} \quad \inc{\vr{x}} \quad \inc{\vr{y}}
\EndLoop
\Loop \label{l:outer}
  \Loop \label{l:upp.b} \label{l:upp.o}
    \State \sub{\vr{c}}{\vr{i}} \quad \inc{\vr{c}'} \label{l:c}
    \Loopatmost{\vr{b}} \label{l:upp.i}
      \State \sub{\vr{d}}{\vr{i}} \quad \sub{\vr{x}}{\vr{i}} \quad
             \add{\vr{d}'}{\vr{i}+1} \label{l:ddu}
    \EndLoop
  \EndLoop
  \Loop \label{l:bu}
    \State \dec{\vr{b}} \quad \add{\vr{b}'}{\vr{i}+1} \label{l:upp.e}
  \EndLoop
  \Loop \label{l:low.b} \label{l:bl}
    \State \dec{\vr{b}'} \quad \inc{\vr{b}}
  \EndLoop
  \Loop \label{l:low.o}
    \State \dec{\vr{c}'} \quad \inc{\vr{c}}
    \Loopatmost{\vr{b}} \label{l:low.i}
      \State \dec{\vr{d}'} \quad
             \inc{\vr{d}} \quad \inc{\vr{x}} \label{l:low.e} \label{l:ddl}
    \EndLoop
  \EndLoop
  \State \inc{\vr{i}} \label{l:next}
\EndLoop
\State \testm{\vr{i}}
\Loop \label{l:f.l}
  \State \sub{\vr{x}}{\vr{i}} \quad \dec{\vr{y}}
\EndLoop
\State \haltz{\vr{y}}
\end{algorithmic}
\end{algorithm}

Observe that the untested counters of program~$\mathcal{F}$ are $\vr{b}$, $\vr{b}'$, $\vr{c}$, $\vr{c'}$, $\vr{d}$, $\vr{d}'$, $\vr{x}$ and $\vr{y}$, and the tested ones are $\vr{i}$ and~$\vr{i}'$ (counter $\vr{i}'$ is hidden in the macros).

\subsection{Correctness}

Before proving that, for any positive integer~$k$, the program~$\mathcal{F}$ is a $k$-amplifier by $k!$, which is the main technical argument in the paper, we provide some intuitions:
\begin{itemize}
\item
the counter~$\vr{d}$ is used to preserve the value of $\vr{x}$ at the end of the main loop, since $\vr{x}$ is modified in the final loop;
\item
the counter~$\vr{d}'$ acts as the auxiliary counter for both $\vr{d}$ and $\vr{x}$, so there is no need to have $\vr{x}'$ as well;
\item
the counter~$\vr{c}$ is initialised to the same positive integer~$a$ as $\vr{d}$, $\vr{x}$ and $\vr{y}$, whereas the counter~$\vr{b}$ is initialised to~$1$;
\item
at the start of any iteration of the main loop in a complete run, the invariant $\vr{d} = \vr{c} \cdot \vr{b}$ will hold, and so the first inner loop will divide $\vr{c}$ by~$\vr{i}$ accurately;
\item
in order for the last inner loop to transfer $\vr{d}'$ fully to $\vr{d}$ and $\vr{x}$, the middle two inner loops will necessarily multiply $\vr{b}$ by $\vr{i} + 1$ accurately;
\item
at the end of the main loop, $\vr{d}$, $\vr{c}$ and $\vr{b}$ will have values $a \cdot k$, $a / (k - 1)!$ and $k!$ (respectively), and in particular $a$ is necessarily divisible by $(k - 1)!$.
\end{itemize}

\begin{lemma}
\label{lem:maintech}
For any positive integer~$k$, the program~$\mathcal{F}$ is a
$k$-amplifier by $k!$, i.e.\ the relation it $k$-computes in counters $\vr{b}, \vr{c}, \vr{d}$ is
\[\{\tuple{b, c, d} \;:\; b = k!,\ c > 0,\ d = c \cdot b\}.\]
\end{lemma}

\begin{proof}
We shall be considering $k$-runs of $\mathcal{F}$ whose initial valuation assigns zero to every counter, and which are either halted or blocked at the {\halt} command because $\vr{y}$ is not zero.  In particular, any such run will have completed the main loop, which runs for $\vr{i} = 1, \ldots, k - 1$.  Hence, we can introduce the following notations for counter values during the $i^{\text{th}}$ iteration of the loop, where $\vr{v}$ is any of the counters $\vr{b}$, $\vr{b}'$, $\vr{c}$, $\vr{c}'$, $\vr{d}$,~$\vr{d}'$:
\begin{description}
\item[$\bar{\vr{v}}_{i}$:]
the final value of $\vr{v}$ after lines \ref{l:upp.b}--\ref{l:upp.e};
\item[$\vr{v}_{i}$:]
the final value of $\vr{v}$ after lines \ref{l:low.b}--\ref{l:low.e}.
\end{description}
It will also be convenient to write $\vr{v}_0$ for the value of $\vr{v}$ at the start of the first iteration of the main loop.  We emphasise that these notations are relative to the run under consideration, which for readability is not written explicitly.

The proof works for any positive integer~$k$ and consists of two parts, that establish the two inclusions between the relation $k$-computed by $\mathcal{F}$ in counters $\vr{b}, \vr{c}, \vr{d}$ and the relation in the statement of the lemma.

The first part, where we assume $b = k!$, $c > 0$ and $d = c \cdot b$, and argue that $\mathcal{F}$ has a complete $k$-run whose final values of counters $\vr{b}, \vr{c}, \vr{d}$ are exactly $b, c, d$, is the easier part.

\begin{claim}
\label{cl:1}
For any $a$ divisible by $(k-1)!$, the program $\mathcal{F}$ has a complete $k$-run which satisfies the equalities in Table~\ref{t:c.r}.
\end{claim}

\begin{table}[h]
\caption{Equalities for counter values in complete $k$-runs of program~$\mathcal{F}$, for all $i = 0, \ldots, k - 1$.}
\label{t:c.r}
\[\begin{array}{rcl@{\qquad}rcl@{\qquad}rcl}
\vr{b}_{0} & = & 1 &
\vr{c}_{0} & = & a &
\vr{d}_{0} & = & a \\
\vr{b}'_{0} & = & 0 &
\vr{c}'_{0} & = & 0 &
\vr{d}'_{0} & = & 0 \\
\bar{\vr{b}}_{i} & = & 0 &
\bar{\vr{c}}_{i} & = & 0 &
\bar{\vr{d}}_{i} & = & 0 \\
\bar{\vr{b}}'_{i} & = & \vr{b}_{i - 1} \cdot (i+1) &
\bar{\vr{c}}'_{i} & = & \vr{c}_{i - 1} / i &
\bar{\vr{d}}'_{i} & = & \vr{d}_{i - 1} \cdot (i+1) / i \\
\vr{b}_{i} & = & \bar{\vr{b}}'_{i} &
\vr{c}_{i} & = & \bar{\vr{c}}'_{i} &
\vr{d}_{i} & = & \bar{\vr{d}}'_{i} \\
\vr{b}'_{i} & = & 0 &
\vr{c}'_{i} & = & 0 &
\vr{d}'_{i} & = & 0
\end{array}\]
\end{table}


\begin{proof}[Proof of Claim~\ref{cl:1}]
Such a run can be built by iterating each inner nondeterministic loop the maximum number of times.  Namely, during iteration $i$ of the main loop:
\begin{itemize}
\item
the loop at line~\ref{l:upp.o} is iterated $\vr{c}_{i - 1} / i$ times and in each pass the loop at line~\ref{l:upp.i} is iterated $\vr{b}_{i - 1}$ times;
\item
the loop at line~\ref{l:bu} is iterated $\vr{b}_{i - 1}$ times;
\item
the loop at line~\ref{l:bl} is iterated $\bar{\vr{b}}_{i}$ times;
\item
the loop at line~\ref{l:low.o} is iterated $\bar{\vr{c}}'_{i}$ times and in each pass the loop at line~\ref{l:low.i} is iterated $\vr{b}_{i}$ times.
\end{itemize}
The divisibility of $a$ by $(k-1)!$ ensures that all divisions in the statement of the claim yield integers.

To see that the run thus obtained can be completed, observe that from the equalities in Table~\ref{t:c.r} it follows that
\begin{align*}
\vr{b}_{k-1} = \prod_{i = 1}^{k-1} \left(i+1\right) = k! \qquad
\vr{c}_{k-1} = a \cdot \prod_{i = 1}^{k-1} \frac{1}{i} = \frac{a}{(k-1)!} \qquad
\vr{d}_{k-1} = a \cdot \prod_{i = 1}^{k-1} \frac{i+1}{i} = a \cdot k. 
\end{align*}
In particular, at the start of the final loop (at line~\ref{l:f.l}), counter $\vr{x}$ equals counter $\vr{d}$ and hence has value $a \cdot k$, and counter $\vr{y}$ has value~$a$.  Iterating the final loop $a$ times therefore reduces $\vr{y}$ (and~$\vr{x}$) to zero as required.
\end{proof}

To obtain $b, c, d$ as the final values of counters $\vr{b}, \vr{c}, \vr{d}$, we apply Claim~\ref{cl:1} with $a = c \cdot (k-1)!$.

We now turn to the remaining second part of the proof of the lemma, where we consider any complete $k$-run and need to show that the final values $b, c, d$ of counters $\vr{b}, \vr{c}, \vr{d}$ satisfy $b = k!$, $c > 0$ and $d = c \cdot b$.

\begin{claim}
\label{cl:2}
For all $i = 1, \ldots, k-1$, we have:
\begin{itemize}
\item
$\bar{\vr{d}}_{i} + \bar{\vr{d}}'_{i} \leq (\vr{d}_{i - 1} + \vr{d}'_{i - 1}) \cdot (i+1) / i$;
\item
$\bar{\vr{d}}_{i} + \bar{\vr{d}}'_{i} =    (\vr{d}_{i - 1} + \vr{d}'_{i - 1}) \cdot (i+1) / i$
if and only if $\bar{\vr{d}}_{i} = \vr{d}'_{i - 1} = 0$;
\item
$\vr{d}_{i} + \vr{d}'_{i} = \bar{\vr{d}}_{i} + \bar{\vr{d}}'_{i}$.
\end{itemize}
\end{claim}

\begin{proof}[Proof of Claim~\ref{cl:2}]
Straightforward calculation based on $(i+1) / i > 1$.
\end{proof}

Let $a$ denote the value of counters $\vr{c}$, $\vr{d}$, $\vr{x}$ and $\vr{y}$ at the start of the main loop.

\begin{claim}
\label{cl:3}
The equalities in Table~\ref{t:c.r} for the values of counters $\vr{d}$ and $\vr{d}'$ are satisfied.
\end{claim}

\begin{proof}[Proof of Claim~\ref{cl:3}]
First, recall that at the start of the final loop (at line~\ref{l:f.l}) counters $\vr{x}$ and $\vr{d}$ are equal, and by Claim~\ref{cl:2} they have value at most $a \cdot k$.  Since counter~$\vr{y}$ has value $a$ at that point and the run is complete, it must actually be the case that the value of $\vr{x}$ here equals $a \cdot k$.  By Claim~\ref{cl:2} again, we infer that for all $i = 1, \ldots, k-1$, we indeed have:
\[\bar{\vr{d}}_{i} = 0 \qquad
  \bar{\vr{d}}'_{i} = \vr{d}_{i - 1} \cdot \frac{i+1}{i} \qquad
  \vr{d}_{i} = \bar{\vr{d}}'_{i} \qquad
  \vr{d}'_{i} = 0. \qedhere\]
\end{proof}

\begin{claim}
\label{cl:4}
We have that $a$ is divisible by $(k-1)!$ and that the equalities in Table~\ref{t:c.r} for the values of counters $\vr{b}$, $\vr{b}'$, $\vr{c}$ and $\vr{c}'$ are satisfied.
\end{claim}

\begin{proof}[Proof of Claim~\ref{cl:4}]
That $a$ is divisible by $(k-1)!$ will follow once we establish the equalities for the values of $\vr{c}$ and $\vr{c}'$, since they involve dividing $a$ by $(k-1)!$.

For the rest of the claim, we argue inductively, where the hypothesis is that the equalities for the values of $\vr{b}$, $\vr{b}'$, $\vr{c}$ and $\vr{c}'$ are satisfied for all indices less than $i$.  Consequently, recalling Claim~\ref{cl:3}, we have that 
\begin{align} \label{eq:indhip}
\vr{d}_{i - 1} = \vr{c}_{i - 1} \cdot \vr{b}_{i - 1}.
\end{align}

Consider the iteration $i$ of the main loop.  We infer from Claim~\ref{cl:3} that the commands in line~\ref{l:ddu} must have been performed $\vr{d}_{i - 1} / i$ times.  Hence, as the values of counters $\vr{b}$ and $\vr{b'}$ remain unchanged until line~\ref{l:bu}, using equation~\eqref{eq:indhip} we deduce that the commands in line~\ref{l:c} must have been performed $\vr{c}_{i-1} / i$ times, and we have:
\[\bar{\vr{c}}_{i} = 0 \qquad
  \bar{\vr{c}}'_{i} = \vr{c}_{i - 1} / i.\]
Also by Claim~\ref{cl:3}, the commands in line~\ref{l:ddl} must have been performed $\bar{\vr{d}}'_{i} = \vr{d}_{i- 1} \cdot (i+1)/i$ times.  From what we have just shown, that number equals $\bar{\vr{c}}'_{i} \cdot \vr{b}_{i - 1} \cdot (i+1)$, and so we conclude that indeed:
\[\begin{array}[b]{rcl@{\qquad}rcl@{\qquad}rcl}
\bar{\vr{b}}_{i} & = & 0 &
\vr{b}_{i} & = & \bar{\vr{b}}'_{i} &
\vr{c}_{i} & = & \bar{\vr{c}}'_{i} \\
\bar{\vr{b}}'_{i} & = & \vr{b}_{i - 1} \cdot (i+1) &
\vr{b}'_{i} & = & 0 &
\vr{c}'_{i} & = & 0.
\end{array}\qedhere\]
\end{proof}

As in the first part, we now conclude that the final values $b, c, d$ of counters $\vr{b}, \vr{c}, \vr{d}$ are $k!, a/(k-1)!, a \cdot k$, and in particular $c \cdot b = a \cdot k! / (k-1)! = d$.
\end{proof}

\section{Main Result}
\label{s:main}

As already indicated, the bulk of the work for our headline result,
namely \tower-hardness of the reachability problem for Petri nets, is
showing how to construct an amplifier without tested counters and for a ratio which is a tower of exponentials.  Most of the pieces have already been developed in Sections \ref{s:elim.zero} and \ref{s:exp}, and here we put them together to obtain a linear-time construction (although any elementary complexity of the reduction would suffice for the \tower-hardness).

\begin{lemma}
\label{lm:main}
An amplifier by $3!^n$ without tested counters is computable in time $O(n)$.
\end{lemma}

\begin{proof}
Letting $\mathcal{A}$ be a trivial amplifier by~$3$ (cf.\ Example~\ref{ex:ratio}), the program
\[\overbrace{\comp{(\comp{(\comp{\mathcal{A}}{\mathcal{F}})}{\mathcal{F}})}{\cdots \mathcal{F}}}^{n \text{ compositions}}\]
is an amplifier by $3!^n$ without tested counters by Proposition~\ref{pr:comp.corr} and Lemma~\ref{lem:maintech}, and it is computable in time $O(n)$ by the definition of the composition operator (cf.\ Section~\ref{s:elim.zero}).
\end{proof}

\begin{theorem} \label{thm:main}
The Petri nets reachability problem is \tower-hard.
\end{theorem}

\begin{proof}
We reduce in linear time from the \tower-complete halting problem for counter programs of size~$n$ with all counters tested and bounded by $3!^n$ (cf.\ Section~\ref{s:c.p}).

Let $\mathcal{M}$ be such a program, and let $\mathcal{T}$ be an
amplifier by $3!^n$ without tested counters which is computable in time $O(n)$ by Lemma~\ref{lm:main}. We have that the composite program $\comp{\mathcal{T}}{\mathcal{M}}$ is without tested counters, and that by Proposition~\ref{pr:comp.corr} it has a complete run if and only if the given program~$\mathcal{M}$ does.
\end{proof}

\begin{corollary}
\label{c:h.expspace}
For any positive integer~$h$, the Petri nets reachability problem with $h + 13$ counters is $h$-\expspace-hard.%
\footnote{We remark that, in the terminology of the classical definition of Petri nets~\cite{Petri62}, the number of places will be $h + 16$ due to $3$ extra places for encoding the control of counter programs.}
\end{corollary}

\begin{proof}
We reduce in linear time from the $h$-\expspace-complete halting problem for counter programs of size~$n$ with $3$ counters, which are all tested and bounded by $n!^{h + 1}$ (cf.\ \cite[Theorems 3.1 and 4.3]{FischerMR68}).

The reduction builds on the following refinement of Lemma~\ref{lm:main}, whose proof is given in Appendix~\ref{app:h.expspace}.
\begin{lemma}
\label{lm:ref}
For every $h\geq 0$, an amplifier by $n!^{h+1}$ without tested counters is computable in time $O(n+h)$, such that:
\begin{itemize}
\item it has $h+13$ untested counters, 
\item $h+1$ counters are required to be zero by the terminal {\halt} command,
\item $9$ out of the $12$ counters not appearing in the {\halt} command are zero at termination of every complete run
of the amplifier.
\end{itemize}
\end{lemma}
\noindent
According to the last condition, the nine counters are \emph{forced} to be zero at termination of every complete run, without being tested to be so. 

Given a counter program $\mathcal{M}$ of size~$n$ with $3$ counters, which are all tested and bounded by $n!^{h + 1}$, the reduction builds the composite program $\comp{\mathcal{T}}{\mathcal{M}}$ where the amplifier $\mathcal{T}$ is given by Lemma~\ref{lm:ref}, analogously as in the proof of Theorem~\ref{thm:main}.
In order to keep the number of counters in $\comp{\mathcal{T}}{\mathcal{M}}$ not greater than $h+13$,
we reuse $6$ out of the $9$ counters not appearing in the {\halt} command of $\mathcal{T}$ (and forced to be zero at termination of $\mathcal{T}$) for simulation of the three counters of 
$\mathcal{M}$.
\end{proof}

\section{Concluding Remarks}

We have focussed on presenting clearly the result that the Petri nets reachability problem is not elementary, leaving several arising directions for future consideration.  The latter include investigating implications for the reachability problem for fixed-dimension flat vector addition systems with states (cf.\ \cite{LerouxS04,BlondinFGHM15,EnglertLT16}).

\section*{Acknowledgements}

We thank Alain Finkel for promoting the study of one-dimensional Petri nets with a pushdown stack, and Matthias Englert, Piotr Hofman and Rados{\l}aw Pi{\'o}rkowski for working with us on open questions about those systems.  The latter efforts led us to discovering the non-elementary lower bound presented in this paper.

We are also grateful to Marthe Bonamy, Artur Czumaj, Javier Esparza, Marcin Pilipczuk, Micha{\l} Pilipczuk, Sylvain Schmitz and Philippe Schnoebelen for helpful comments.

\bibliographystyle{plainnat}
\bibliography{tower}

\clearpage
\appendix

\section{Proof of Lemma~\ref{lm:ref}}
\label{app:h.expspace}.
\begin{proof}
Let $\mathcal{A}$ be the program obtained from a trivial amplifier by~$n$ (cf.\ Example~\ref{ex:ratio}):
\begin{quote}
\begin{algorithmic}[1]
\State \add{\vr{b}_0}{n}
  \State \inc{\vr{c}_0} \quad \add{\vr{d}_0}{n}
\Loop
  \State \inc{\vr{c}_0} \quad \add{\vr{d}_0}{n}
\EndLoop
\State \halt.
\end{algorithmic}
\end{quote}
and let $\mathcal{F}$ be the factorrial amplifier from Lemma~\ref{lem:maintech}.
By Proposition~\ref{pr:comp.corr} and Lemma~\ref{lem:maintech}, we know that the composite program:
\[ \mathcal{T} \ = \ \overbrace{\comp{(\comp{(\comp{\mathcal{A}}{\mathcal{F}})}{\mathcal{F}})}{\cdots \mathcal{F}}}^{h+1 \text{ compositions}} \]
is an amplifier by $n!^{h+1}$ without tested counters.
By expanding the $h+1$ composition operations (cf.~Section~\ref{s:elim.zero}) and renaming counters explicitly (the counters of $\mathcal{A}$ are indexed by $0$, and counters of the $j$th program $\mathcal{F}$, for $j = 1, \ldots, h+1$, are indexed by $j$), we obtain the following form of $\mathcal{T}$:
\begin{quote}
\begin{algorithmic}
\State $\mathcal{A}'$ \quad $\mathcal{H}_1$ \quad $\cdots$ \quad $\mathcal{H}_{h + 1}$
\State \haltz{\vr{d}_0, \vr{y}_1, \vr{d}_1, \ldots, \vr{y}_h, \vr{d}_h, \vr{y}_{h+1}}
\end{algorithmic}
\end{quote}
where $\mathcal{A}'$ is the program fragment of $\mathcal{A}$ with the {\halt} command removed, and $\mathcal{H}_j$ is the program fragment defined in Algorithm~\ref{a:Hj} 
(without the boxed commands, which will be used later to simplify the program) making
use of the following macros depending on~$j$.
We omit the \iszero{\vr{i}_j} and \iszero{\vr{i}'_j} macros which are defined analogously to \ismax{\vr{i}_j} (cf.~Section~\ref{s:elim.zero}), and \sub{\vr{c}_j}{\vr{i}_j}, \sub{\vr{d}_j}{\vr{i}_j}, \add{\vr{d}'_j}{\vr{i}_j+1}
and \add{\vr{b}'_j}{\vr{i}_j+1} macros which are defined analogously
to \sub{\vr{x}_j}{\vr{i}_j} (cf.~Section~\ref{subs:simple}).
\begin{description}
\item[\setup{\vr{\hat i}_j, \vr{\hat i}'_j}:] \ (cf.~\eqref{e:comp.2} in Section~\ref{s:elim.zero})
\begin{samepage}
\begin{quote}
\begin{algorithmic}
\Loop
  \State \inc{\vr{\hat i}_j} \quad \inc{\vr{\hat i}'_j} \quad \dec{\vr{b}_{j-1}} \quad \dec{\vr{d}_{j - 1}}
\EndLoop
\State \dec{\vr{c}_{j - 1}}
\end{algorithmic}
\end{quote}
\end{samepage}

\item[\ismax{\vr{i}_j}:] \ (cf.~\eqref{e:comp.5z} and~\eqref{e:comp.5m} in Section~\ref{s:elim.zero})
\begin{quote}
\begin{algorithmic}
\Loop
  \State \dec{\vr{i}_j} \quad \inc{\vr{\hat i}_j} \quad \dec{\vr{d}_{j - 1}}
\EndLoop
\State \dec{\vr{c}_{j - 1}}
\Loop
  \State \inc{\vr{i}_j} \quad \dec{\vr{\hat i}_j} \quad \dec{\vr{d}_{j - 1}}
\EndLoop
\State \dec{\vr{c}_{j - 1}}
\end{algorithmic}
\end{quote}

\item[\sub{\vr{x}_j}{\vr{i}_j}:] \ (cf.~Section~\ref{subs:simple})
\begin{quote}
\begin{algorithmic}
\Loop
  \State \dec{\vr{i}_j} \quad \inc{\vr{\hat i}_j} \quad \inc{\vr{i}'_j} \quad \dec{\vr{x}_j} 
\EndLoop
\State \iszero{\vr{i}_j}
\Loop
  \State  \inc{\vr{i}_j} \quad \dec{\vr{\hat i}_j} \quad \dec{\vr{i}'_j}
\EndLoop
\State \iszero{\vr{i}'_j}
\end{algorithmic}
\end{quote}

\item[\textbf{loop at most} $\vr{b}_j$ \textbf{times} $<${\it body}$>$:] \ (cf.~Section~\ref{s:exp})
\begin{quote}
\begin{algorithmic}
\Loop
  \State \dec{\vr{b}_j} \quad \inc{\vr{b}'_j}
\EndLoop
\Loop
  \State \dec{\vr{b}'_j} \quad \inc{\vr{b}_j}
  \State $<${\it body}$>$
\EndLoop
\end{algorithmic}
\end{quote}
\end{description}

\begin{algorithm}
\caption{Program fragment~$\mathcal{H}_j$ and additional (boxed) commands extending it to $\mathcal{H}^{(1)}_j$.}
\label{a:Hj}
\begin{algorithmic}[1]
\State \setup{\vr{\hat i}_j, \vr{\hat i}'_j}
\State \inc{\vr{i}_j} \quad \dec{\vr{\hat i}_j}  \hfill (cf.~\eqref{e:comp.3} and~\eqref{e:comp.4} in Section~\ref{s:elim.zero})
\State \inc{\vr{b}_j} \quad
	 \inc{\vr{c}_{j}} \quad
       \inc{\vr{d}_j} \quad \inc{\vr{x}_j} \quad
\label{l:insertinc1} 
       \inc{\vr{y}_j}
	\quad \fbox{ \inc{\vr{d}_{j - 1}} }
\Loop
  \State \inc{\vr{c}_{j}} \quad
         \inc{\vr{d}_j} \quad \inc{\vr{x}_j} \quad
\label{l:insertinc2}
         \inc{\vr{y}_j}
         \quad \fbox{ \inc{\vr{d}_{j - 1}} }
\EndLoop
\Loop
  \Loop
    \State \sub{\vr{c}_{j}}{\vr{i}_j} \quad \inc{\vr{c}'_j}
    \Loopatmost{\vr{b}_j}
      \State \sub{\vr{d}_j}{\vr{i}_j} \quad \sub{\vr{x}_j}{\vr{i}_j} \quad
             \add{\vr{d}'_j}{\vr{i}_j+1} 
    \EndLoop
  \EndLoop
  \Loop
    \State \dec{\vr{b}_j} \quad \add{\vr{b}'_j}{\vr{i}_j+1}
  \EndLoop
  \Loop
    \State \dec{\vr{b}'_j} \quad \inc{\vr{b}_j}
  \EndLoop
  \Loop
    \State \dec{\vr{c}'_j} \quad \inc{\vr{c}_{j}}
    \Loopatmost{\vr{b}_j}
      \State \dec{\vr{d}'_j} \quad 
             \inc{\vr{d}_j} \quad \inc{\vr{x}_j}
    \EndLoop
  \EndLoop
  \State \inc{\vr{i}_j} \quad \dec{\vr{\hat i}_j}   \hfill (cf.~\eqref{e:comp.3} and~\eqref{e:comp.4} in Section~\ref{s:elim.zero})
\EndLoop
\State \ismax{\vr{i}_j}
\Loop
  \State \sub{\vr{x}_j}{\vr{i}_j} \quad 
  \dec{\vr{y}_j} \label{l:insertdec}
  \quad \fbox{ \dec{\vr{d}_{j - 1}} }
\EndLoop
\State \fbox{ \reset{\vr{i}_j, \vr{\hat i}'_j} }  \label{l:insertreset}
\end{algorithmic}
\end{algorithm}

We thus have:
\begin{claim}
The program $\mathcal{T}$ is an amplifier by $n!^{h+1}$ without tested counters.
\end{claim}

We call the counters $\vr{b}_j, \vr{c}_j, \vr{d}_j$, for $j = 0, \ldots, h+1$, the \emph{ratio} counters, and the remaining counters \emph{non-ratio} ones.
The following is a straightforward observation from the above explicit construction of 
the amplifier $\mathcal{T}$:
\begin{claim}
For $j = 1, \ldots, h+1$, the program fragment $\mathcal{H}_j$ only modifies the $j$-indexed counters, plus 
the three ratio counters  $\vr{b}_{j-1}, \vr{c}_{j-1}, \vr{d}_{j-1}$.
\end{claim}
We are going to perform a sequence of optimisations on $\mathcal{T}$
leading to an amplifier that satisfies
the requirements of Lemma~\ref{lm:ref}.

The first optimisation builds on the following fact, which is easily derived from
the analysis performed in the proof of Lemma~\ref{lem:maintech}:
\begin{claim} \label{claim:zeroattermination}
Let $k > 0$. At termination of every complete $k$-run of $\mathcal{F}$, the counters
$\vr{x}, \vr{d}', \vr{c}', \vr{b'}, \vr{i}'$ are all zero, and the counter $\vr{i}$ equals $k$.
\end{claim}

Consider a complete run of $\mathcal{T}$ and fix $j \in \{1, \ldots, h+1\}$. To analyse the program fragment~$\mathcal{H}_j$, it is convenient to denote by $k$ the value of $\vr{b}_{j-1}$ at the start of $\mathcal{H}_j$. 
By Claim~\ref{claim:zeroattermination}, at the end of the program fragment $\mathcal{H}_j$, not only the 
counter $\vr{y}_j$ (which is later checked to be zero by the {\halt} command) but actually
all $j$-indexed non-ratio counters are zero, except for the two counters $\vr{i}_j, \vr{\hat i}'_j$.
It is readily verified that these two counters will be equal to $k$ 
due to the invariants we keep in the program: $\vr{i} + \vr{\hat i} = k$ and $\vr{i}' + \vr{\hat i}' = k$.
We enforce the counters $\vr{i}_j, \vr{\hat i}'_j$ to be zero at the end of $\mathcal{H}_j$, by adjoining at the end of 
$\mathcal{H}_j$ a macro defined by the following
piece of code (depicted as a boxed command in line~\ref{l:insertreset} Algorithm~\ref{a:Hj}):
\begin{description}
\begin{samepage}
\item[\reset{\vr{i}_j, \vr{\hat i}'_j}:] \
\begin{quote}
\begin{algorithmic}
\Loop
  \State \dec{\vr{i}_j} \quad \dec{\vr{\hat i}'_j} \quad \dec{\vr{d}_{j - 1}}
\EndLoop
\State \dec{\vr{c}_{j - 1}}
\end{algorithmic}
\end{quote}
\end{samepage}
\end{description}
Denoting by $\mathcal{H}^{(0)}_j$ the so modified program fragments, and by $\mathcal{T}^{(0)}$ the corresponding program
\begin{quote}
\begin{algorithmic}
\State $\mathcal{A}'$ \quad $\mathcal{H}^{(0)}_1$ \quad $\cdots$ \quad $\mathcal{H}^{(0)}_{h + 1}$
\State \haltz{\vr{d}_0, \vr{y}_1, \vr{d}_1, \ldots, \vr{y}_h, \vr{d}_h, \vr{y}_{h+1}},
\end{algorithmic}
\end{quote}
we summarise:
\begin{claim} \label{claim:nonratiozero}
For $j = 1, \ldots, h+1$, at the end of the program fragment $\mathcal{H}^{(0)}_j$ in every complete run of 
$\mathcal{T}^{(0)}$, all $j$-indexed non-ratio counters are zero.
\end{claim}

In the next optimisation, we reduce the number of counters that are required to be zero by the terminal command
{\halt} from
$2h+2$ to $h+1$. 
Consider the following further modification of $\mathcal{H}^{(0)}_j$ depicted by boxed commands in 
lines~\ref{l:insertinc1}, \ref{l:insertinc2} and~\ref{l:insertdec} Algorithm~\ref{a:Hj}:
\begin{enumerate}
\item\label{en:1} insert the instruction \inc{\vr{d}_{j-1}} in lines~\ref{l:insertinc1} and~\ref{l:insertinc2}
(which results in additional increasing the value of $\vr{d}_{j-1}$ by the value ultimately achieved by $\vr{y}_j$), and
\item\label{en:2} insert the instruction \dec{\vr{d}_{j-1}} in line~\ref{l:insertdec} (since it is simultaneously decreased with $\vr{y}_j$ it results in decreasing the value of
$\vr{d}_{j-1}$ by at most the value of the additional increase).
\end{enumerate}
We denote the resulting program fragment by $\mathcal{H}^{(1)}_j$.

Observe that in complete runs of $\mathcal{T}$ the decrease of $\vr{d}_{j-1}$ in~\eqref{en:2} being equal to the increase in~\eqref{en:1},
is equivalent to $\vr{y}_j$ being zero at the end of $\mathcal{H}^{(1)}_j$.
First, let us shortly comment that this does not harm the correctness of zero tests.
Indeed, it suffices to replace the invariant $\vr{d}_{j-1} \ge \vr{c}_{j-1} \cdot R$ used in the proof of Proposition~\ref{pr:comp.corr} with the invariant $\vr{d}_{j-1} \ge \vr{c}_{j-1} \cdot R + \vr{y}_{j-1}$.
The point of this construction is to show that we can remove all $\vr{y}_j$ from the final {\halt} command.
Indeed, using the new invariant $\vr{d}_{j-1} \ge \vr{c}_{j-1} \cdot R + \vr{y}_{j-1}$ we have $\vr{d}_{j-1} \ge \vr{y}_{j-1}$.
Therefore, the following program
$\mathcal{T}^{(1)}$ that requires only counters $\vr{d}_0$, \ldots, $\vr{d}_h$ to be zero at termination:
\begin{samepage}
\begin{quote}
\begin{algorithmic}
\State $\mathcal{A}'$ \quad $\mathcal{H}^{(1)}_1$ \quad $\cdots$ \quad $\mathcal{H}^{(1)}_{h + 1}$
\State \haltz{\vr{d}_0, \vr{d}_1, \ldots, \vr{d}_h}
\end{algorithmic}
\end{quote}
\end{samepage}
is an amplifier by $n!^{h+1}$, just as well as $\mathcal{T}$.
In other words, for every $j \in \{1, \ldots, h+1\}$, 
the counter $\vr{y}_j$ is \emph{forced} to be zero at the end of $\mathcal{H}^{(1)}_j$ in every complete run
of $\mathcal{T}^{(1)}$.  We have thus obtained:
\begin{claim}
The program $\mathcal{T}^{(1)}$ is an amplifier by $n!^{h+1}$ without tested counters.
\end{claim}

Since Claim~\ref{claim:nonratiozero} is still true for the program $\mathcal{T}^{(1)}$, we 
optimise further by collapsing all the counters $\vr{x}_1$, \ldots, $\vr{x}_{h+1}$ into one counter $\vr{x}$,
and analogously for all other non-ratio counters. 
In particular, since the counters $\vr{y}_1$, \ldots, $\vr{y}_{h+1}$ do not appear anymore in the terminal {\halt} command
of $\mathcal{T}^{(1)}$, we replace these counters with one counter $\vr{y}$. 
We thus obtain new program fragments 
$\mathcal{H}^{(2)}_j$, defined in Algorithm~\ref{a:H2j}, each of them using the same 9 non-ratio counters 
$\vr{i}, \vr{\hat i}, \vr{i}', \vr{\hat i}', \vr{b}', \vr{c}', \vr{d}', \vr{x}, \vr{y}$, plus six ratio counters 
$\vr{b}_{j-1}, \vr{c}_{j-1}, \vr{d}_{j-1}, \vr{b}_j, \vr{c}_j, \vr{d}_j$ (the macros used in $\mathcal{H}^{(2)}_j$ are adjusted accordingly).
The optimisation results in a new program $\mathcal{T}^{(2)}$:
\begin{quote}
\begin{algorithmic}
\State $\mathcal{A}'$ \quad $\mathcal{H}^{(2)}_1$ \quad $\cdots$ \quad $\mathcal{H}^{(2)}_{h + 1}$
\State \haltz{\vr{d}_0, \vr{d}_1, \ldots, \vr{d}_h}
\end{algorithmic}
\end{quote}
satisfying the following claim:
\begin{claim}
The program $\mathcal{T}^{(2)}$ is an amplifier by $n!^{h+1}$ without tested counters such that:
\begin{itemize}
\item it has $3(h+2)+9 = 3h + 15$ untested counters, 
\item $h+1$ counters are required to be zero by the terminal {\halt} command,
\item all counters except for $\vr{b}_{h+1}, \vr{c}_{h+1}, \vr{d}_{h+1}$ are zero at termination of every complete run.
\end{itemize}
\end{claim}
\noindent
Indeed, concerning the last condition, at the end of the program fragment $\mathcal{H}^{(2)}_j$
in a complete run, 
the counter $\vr{d}_{j-1}$ is necessarily zero (as at termination it is so), and hence also $\vr{c}_{j-1}$ is necessarily zero. 
Furthermore, the analysis in the proof of Lemma~\ref{lem:maintech} reveals that $\vr{b}_{j-1}$ is zero too.
In consequence, all the counters
$\vr{b}_0$, \ldots, $\vr{b}_h$, $\vr{c}_0$, \ldots, $\vr{c}_h$ and $\vr{d}_0$, \ldots, $\vr{d}_h$ are necessarily zero at the termination 
of every complete run.

\begin{algorithm}
\caption{Program fragment~$\mathcal{H}^{(2)}_j$.}
\label{a:H2j}
\begin{algorithmic}[1]
\State \setup{\vr{\hat i}, \vr{\hat i}'} \label{lH2:1}
\State \inc{\vr{i}} \quad \dec{\vr{\hat i}}
\State \inc{\vr{b}_j} \quad
	 \inc{\vr{c}_{j}} \quad
       \inc{\vr{d}_j} \quad \inc{\vr{x}} \quad
       \inc{\vr{y}}
	\quad \inc{\vr{d}_{j - 1}} 
\Loop
  \State \inc{\vr{c}_{j}} \quad
         \inc{\vr{d}_j} \quad \inc{\vr{x}} \quad
         \inc{\vr{y}}
         \quad \inc{\vr{d}_{j - 1}} 
\EndLoop
\Loop
  \Loop
    \State \sub{\vr{c}_{j}}{\vr{i}} \quad \inc{\vr{c}'}
    \Loopatmost{\vr{b}_j}
      \State \sub{\vr{d}_j}{\vr{i}} \quad \sub{\vr{x}}{\vr{i}} \quad
             \add{\vr{d}'}{\vr{i}+1} 
    \EndLoop
  \EndLoop
  \Loop
    \State \dec{\vr{b}_j} \quad \add{\vr{b}'}{\vr{i}+1}
  \EndLoop
  \Loop
    \State \dec{\vr{b}'} \quad \inc{\vr{b}_j}
  \EndLoop
  \Loop
    \State \dec{\vr{c}'} \quad \inc{\vr{c}_{j}}
    \Loopatmost{\vr{b}_j}
      \State \dec{\vr{d}'} \quad
             \inc{\vr{d}_j} \quad \inc{\vr{x}}
    \EndLoop
  \EndLoop
  \State \inc{\vr{i}} \quad \dec{\vr{\hat i}}
\EndLoop
\State \ismax{\vr{i}}
\Loop
  \State \sub{\vr{x}}{\vr{i}} \quad 
  \dec{\vr{y}}
  \quad \dec{\vr{d}_{j - 1}} 
\EndLoop
\State \reset{\vr{i}, \vr{\hat i}'} 
\end{algorithmic}
\end{algorithm}

Using the latter observation we perform further (final) optimisations.

First, as every $c_j$ is only used in $\mathcal{H}^{(2)}_j$ and $\mathcal{H}^{(2)}_{j+1}$, and is 
zero at termination
(and thus also at the end of $\mathcal{H}^{(2)}_{j+1}$) in every complete run, instead of $h+2$ counters 
$\vr{c}_0$, \ldots, $\vr{c}_{h+1}$ we can use in an alternating manner just two, denoted $\vr{c}_0$ and $\vr{c}_1$.

Second, as every $\vr{b}_{j-1}$ is zero at termination in every complete run, and not modified after line~\ref{lH2:1} in $\mathcal{H}^{(2)}_j$, it becomes zero before the next counter $\vr{b}_j$ is modified by
$\mathcal{H}^{(2)}_j$.
Therefore, all $h+2$ counters $\vr{b}_0$, \ldots, $\vr{b}_{h+1}$ can be collapsed to just one counter $\vr{b}$.

Applying these two optimisations, 
together with a further improvement in macros that eliminate one more counter 
(to be expanded below),
yields our final version of the program fragment $\mathcal{H}^{(3)}_j$ as defined in Algorithm~\ref{a:H3j}.
The \textbf{loop at most} $\vr{b}$ \textbf{times} macro is exactly the same as in Section~\ref{s:exp},
and the \ismax{\vr{i}} macro appearing in $\mathcal{H}^{(3)}_j$ is as follows:
\begin{description}
\item[\ismax{\vr{i}}:] \ 
\begin{quote}
\begin{algorithmic}
\Loop
  \State \dec{\vr{i}} \quad \inc{\vr{i}'} \quad \dec{\vr{d}_{j - 1}}
\EndLoop
\State \dec{\vr{c}_{j - 1 \bmod 2}}
\Loop
  \State \inc{\vr{i}} \quad \dec{\vr{i}'} \quad \dec{\vr{d}_{j - 1}}
\EndLoop
\State \dec{\vr{c}_{j - 1 \bmod 2}}
\end{algorithmic}
\end{quote}
\end{description}
We implement the \sub{\vr{x}}{\vr{i}} macro succinctly, eliminating the counter $\vr{\hat i}'$,
 as follows:
\begin{description}
\begin{samepage}
\item[\sub{\vr{x}}{\vr{i}}:] \
\begin{quote}
\begin{algorithmic}
\Loop
  \State \dec{\vr{i}} \quad \inc{\vr{i}'} \quad \dec{\vr{x}} \quad \dec{\vr{d}_{j - 1}}
\EndLoop
\Loop
  \State \dec{\vr{\hat i}} \quad \inc{\vr{i}} \quad \dec{\vr{d}_{j - 1}}
\EndLoop
\State \dec{\vr{c}_{j - 1 \bmod 2}}
\Loop
  \State \dec{\vr{i}} \quad \inc{\vr{\hat i}} \quad \dec{\vr{d}_{j - 1}}
\EndLoop
\Loop
  \State \dec{\vr{i'}} \quad \inc{\vr{i}} \quad \dec{\vr{d}_{j - 1}}
\EndLoop
\State \dec{\vr{c}_{j - 1 \bmod 2}}
\end{algorithmic}
\end{quote}
\end{samepage}
\end{description}
\noindent
Analogously, we optimise the other macros 
\sub{\vr{c_0}}{\vr{i}}, \sub{\vr{c_1}}{\vr{i}}, \sub{\vr{d}_j}{\vr{i}}, \add{\vr{d}'}{\vr{i}+1} and \add{\vr{b}'}{\vr{i}+1}
featuring in $\mathcal{H}^{(3)}_j$.

Finally, we deliberately remove $\vr{\hat i}'$ from \setup{\vr{\hat i}} and \reset{\vr{i}} macros, respectively:
\begin{description}
\item[\setup{\vr{\hat i}}:] \  (note that the \inc{\vr{\hat i}'} command is not present)
\begin{quote}
\begin{algorithmic}
\Loop
  \State \inc{\vr{\hat i}} 
  \quad \dec{\vr{b}} \quad \dec{\vr{d}_{j - 1}}
\EndLoop
\State \dec{\vr{c}_{j - 1 \bmod 2}}
\end{algorithmic}
\end{quote}

\item[\reset{\vr{i}}:] \  (note that the \dec{\vr{\hat i}'} command is not present)
\begin{quote}
\begin{algorithmic}
\Loop
  \State \dec{\vr{i}} \quad \dec{\vr{d}_{j - 1}}
\EndLoop
\State \dec{\vr{c}_{j - 1 \bmod 2}}
\end{algorithmic}
\end{quote}
\end{description}

\begin{algorithm}[ht]
\caption{Program fragment~$\mathcal{H}^{(3)}_j$.}
\label{a:H3j}
\begin{algorithmic}[1]
\State \setup{\vr{\hat i}}
\State \inc{\vr{i}} \quad \dec{\vr{\hat i}}
\State \inc{\vr{b}} \quad
	 \inc{\vr{c}_{j \bmod 2}} \quad
       \inc{\vr{d}_j} \quad \inc{\vr{x}} \quad 
       \inc{\vr{y}}
	\quad \inc{\vr{d}_{j - 1}} 
\Loop
  \State \inc{\vr{c}_{j \bmod 2}} \quad
         \inc{\vr{d}_j} \quad \inc{\vr{x}} \quad
         \inc{\vr{y}}
         \quad \inc{\vr{d}_{j - 1}} 
\EndLoop
\Loop
  \Loop
    \State \sub{\vr{c_{j \bmod 2}}}{\vr{i}} \quad \inc{\vr{c}'}
    \Loopatmost{\vr{b}}
      \State \sub{\vr{d}_j}{\vr{i}} \quad \sub{\vr{x}}{\vr{i}} \quad
             \add{\vr{d}'}{\vr{i}+1} 
    \EndLoop
  \EndLoop
  \Loop
    \State \dec{\vr{b}} \quad \add{\vr{b}'}{\vr{i}+1}
  \EndLoop
  \Loop
    \State \dec{\vr{b}'} \quad \inc{\vr{b}}
  \EndLoop
  \Loop
    \State \dec{\vr{c}'} \quad \inc{\vr{c}_{j \bmod 2}}
    \Loopatmost{\vr{b}}
      \State \dec{\vr{d}'} \quad
             \inc{\vr{d}_j} \quad \inc{\vr{x}}
    \EndLoop
  \EndLoop
  \State \inc{\vr{i}} \quad \dec{\vr{\hat i}}
\EndLoop
\State \ismax{\vr{i}}
\Loop
  \State \sub{\vr{x}}{\vr{i}} \quad 
  \dec{\vr{y}}
  \quad \dec{\vr{d}_{j - 1}} 
\EndLoop
\State \reset{\vr{i}} 
\end{algorithmic}
\end{algorithm}

We conclude that:
\begin{samepage}
\begin{claim}
The corresponding program $\mathcal{T}^{(3)}$ defined as
\begin{quote}
\begin{algorithmic}
\State $\mathcal{A}'$ \quad $\mathcal{H}^{(3)}_1$ \quad $\cdots$ \quad $\mathcal{H}^{(3)}_{h + 1}$
\State \haltz{\vr{d}_0, \vr{d}_1, \ldots, \vr{d}_h}
\end{algorithmic}
\end{quote}
is an amplifier by  $n!^{h+1}$  without tested counters satisfying the conditions of Lemma~\ref{lm:ref}. 
\end{claim}
\end{samepage}

The proof of Lemma~\ref{lm:ref} is thus completed.
\end{proof}

\end{document}